\documentclass[preprint,12pt]{elsarticle}

\usepackage{enumitem}
\usepackage{graphicx}
\usepackage{comment}

\usepackage{amsfonts}
\usepackage{amsthm}
\usepackage{amsmath}
\usepackage{amssymb}

\newtheorem{theorem}{Theorem}
\newtheorem{lemma}[theorem]{Lemma}

 \newtheorem{definition}[theorem]{Definition}
\newcommand{\ceiling}[1]{\left\lceil #1 \right\rceil}
\newcommand{\floor}[1]{\left\lfloor #1 \right\rfloor}

\begin{document}

\begin{frontmatter}

\title{Streaming Approximation Scheme for Minimizing Total Completion Time on Parallel Machines
Subject to Varying Processing Capacity
}

\author[inst1]{Bin Fu}
\ead{bin.fu@utrgv.edu}
\affiliation[inst1]{organization = {Department of Computer Science},
addressline = {University of Texas Rio Grande Valley}, city = { Edinburg}, state = {TX}, postcode={78539}, country={USA}}

\author[inst2]{Yumei Huo}
\ead{yumei.huo@csi.cuny.edu}
\affiliation[inst2]{organization = {Department of Computer Science},
addressline = {College of Staten Island, CUNY}, city = { Staten Island}, state = {NY}, postcode={10314}, country={USA}}

\author[inst3]{Hairong Zhao  \corref{cor1}}
\cortext[cor1]{Corresponding author}
\ead{hairong@purdue.edu}
\affiliation[inst3]{organization = {Department of Computer Science},
addressline = {Purdue University Northwest}, city = { Hammond}, state = {IN}, postcode={46323}, country={USA}}

\begin{abstract}
We study the problem of minimizing total completion time on parallel machines subject to varying processing capacity.
%The processing capacity of a machine refers to the availability of this machine for processing jobs, which can be represented by a value in $[0, 1]$, where $1$ and $0$ represent the full availability and zero availability of the machine respectively.
% we allow the processing capacity of machines to 
%Varying processing capacity  is a common situation in service industries.
In this paper, we develop an approximation scheme for the problem under the data stream model where the input data is massive and cannot fit into memory and thus can only be scanned for a few passes.
Our algorithm can compute the approximate value of the optimal total
completion time in one pass and output the schedule with the
approximate value in two passes.
\end{abstract}


\end{frontmatter}

\section{Introduction}
%Varying processing capacity  is a common situation in service industries and has been studied in many different applications. 
In 1980, Baker and Nuttle \cite{bn80} studied the problem of scheduling $n$ jobs that require a single resource whose availability varies over time. This model was motivated  by the situations in which machine availability may be temporarily reduced to conserve energy or interrupted for scheduled maintenance. It also applies to situations in which processing requirements are stated in terms of man-hours and the labor availability varies over time. One example application is rotating Saturday shifts, where a company only maintains a fraction, for example 33\%, of the workforce every Saturday.

In 1987, Adiri and Yehudai \cite{ay87} studied the scheduling problem on single and parallel machines where the service rate of a machine remains constant while a job is being processed and may only be changed upon its completion. A simple application example is a machine tool whose performance is a function of the quality of its cutters which can be replaced only upon completion of one job.

In 2016, Hall et. al \cite{hll16} proposed a new model of multitasking via shared processing which allows a team to continuously work on its main, or primary tasks while a fixed percentage of its processing capacity may be allocated to process the routinely scheduled activities such as administrative meetings, maintenance work or meal breaks.
%In many practical situations of this multitasking scheduling model, the schedule of the routine jobs are usually pre-determined (\cite{hll16}\cite{fhz22}). Thus 
In these scenarios, a working team can be viewed as a machine with reduced  capacity in some periods for processing primary jobs.  A manager needs to decide how to schedule the primary jobs on these shared processing machines so as to optimize some performance criteria. In  \cite{hll16}, the authors studied the single machine environment only. They assumed that the processing capacity allocated to all the routine jobs is a constant $e$ that is strictly less than 1.

Similar models also occur in queuing system where the number of servers can change, or where the service rate of each server can change. In \cite{t86}, Teghem defined these models as the vacation models and the variable service rate models, respectively.  As Doshi pointed out in \cite{d86}, queuing systems where the server works on primary and secondary customers arise naturally in many computer, communication and production systems. From the point of the primary customers' view, the server working on the secondary customers is equivalent to the server taking a vacation and not being available to the primary customers during this period.

In this paper, we extend the research on scheduling subject to varying machine processing capacity and study the problems under parallel machine environment so as to optimize some objectives. 
In our model, we allow different processing capacity during different periods and the change of the processing capacity is independent of the jobs.
Although in some aspects the problems have been intensively studied, such as scheduling subject to unavailability constraint, this work is targeting on a more general model compared with all the above mentioned research. %such as those proposed by Baker and Nuttle \cite{bn80}, by Adiri and Yehudai in \cite{ay87}, and by Hall et. al (\cite{hll16}) in 2016. %In this paper, we generalize the work in \cite{hll16} to  parallel machines environment, and we allow different processing capacity being allocated to different routine jobs  and furthermore we allow  processing capacity   allocated to a routine job being changed during its execution.    Our model can also be seen as a generalization of some earlier models such as those proposed by by Baker and Nuttle \cite{bn80}, and by Adiri and Yehudai in \cite{ay87}, respectively.
This generalized model apparently has a lot of applications which have been discussed in the above literature. 
Due to historic reasons and different application contexts, different terms have been used in literature to refer to  similar concepts, including service rate \cite{ay87}\cite{t86}, processing capacity \cite{awkl19}\cite{ccc10}\cite{jkt11}, machine capacity\cite{hll16}, sharing ratio\cite{hll16}, etc. In this paper, we will adopt the term processing capacity to refer to the availability of a machine for processing the jobs.

For the proposed general model, in \cite{fhz22} we have studied some problems under the traditional  data  model  where all data can be stored locally and accessed  in constant time. In this paper, we will study the proposed model under the data stream environment where the input data is massive and cannot be read into memory. Specifically, we study the data stream model of our problem where the number of jobs is so big that jobs' information cannot be stored but can only be scanned in one or more passes. This research, as in many other research areas, caters to the need for providing  solutions under big data that also emerges in the area of scheduling. 

As Muthukrishnan wrote in his paper \cite{m05}, 
%with the traditional way that data are fed from the memory, %and one can modify the underlying data to reflect the updates; 
%real time queries can be done simply by looking up a value in the memory,
%, as we see in banking and credit transactions; as far as 
%while some complex analyses such as trend analysis, forecasting, etc., are usually performed offline. However, 
in the modern world, with more and more data generated, the automatic data feeds are needed for many tasks in the monitoring applications such as atmospheric, astronomical, networking, financial, sensor-related fields, etc. %For example,large amount of data need to be fed and processed in a short time to monitor complex correlations, track trends, support exploratory analyses and perform complex tasks such as classification, harmonic analysis etc. 
These tasks are time critical and thus it is important to process them in almost real time mode so as to accurately keep pace with the rate of stream updates and  reflect rapidly changing trends in the data. 
Therefore, %With more data generated and more demands of data streams processing for now and in the future, 
the researchers are facing the questions under the current and future trend of more demands of data streams processing: what can we (not) do if we are given a certain amount of resources, a data stream rate and a particular analysis task?

%While some methods are available for processing large amount of data of these time critical tasks, such as making things parallel, controlling data rate by sampling or shedding updates, rounding data structures to certain block boundaries, using hierarchically detailed analysis, etc., these approaches are ultimately limiting. 

A natural approach to dealing with large amount of data of these time critical tasks involves approximations and developing data stream algorithms. Streaming algorithms were initially studied by Munro and Paterson in 1978 (\cite{mp78}), and then by Flajolet and Martin in 1980s (\cite{fm85}). The model was formally established by Alon, Matias, and Szegedy in ~\cite{ams99} and has received a lot of attention since then. 

In this work our goal is to design streaming algorithms for our proposed problem to approximate the optimal solution in a limited number of passes over the data and using limited space. 
Formally, streaming algorithms are algorithms for processing the input where some or all of of the data is not available for random access but rather arrives as a sequence of items and can be examined in only a few passes (typically just one). 
The performance of streaming algorithms is measured by three factors: the number of passes the algorithm must run over the stream, the space needed and the update time of the algorithm.

%If a machine has capacity $\alpha$, $1 \le \alpha \le 1$ during an interval of length $l$, then  the largest job that can be completed during this interval is $\alpha \cdot l$. Thus a job may have to be processed in several continuous intervals.  See the figure???
%For a processing capacity It is a positive in $(0,1]$, at this rate, the amount of work that can be processed during an interval of length $l$ is $\alpha_j l$.

\subsection{Problem Definition}

Formally our scheduling model can be defined as follows. There is a set $N = \{1, \cdots, n\}$ of $n$ jobs and $m$ parallel machines where the processing capacity of machines varies over time. Each job $j \in N $ has a processing time $p_j$ and can be processed by any one of the machines uninterruptedly.
Associated with each machine $M_i$ are $l_i$ continuous intervals during which the  processing capacity of $M_i$ are $\alpha_{i,1}$, $\alpha_{i, 2}$, $\ldots$, $\alpha_{i,l_i}$, respectively, see Figure~\ref{fig:varying-processing-capacity} for an example.
\begin{figure}[h]
\includegraphics[width=\textwidth]{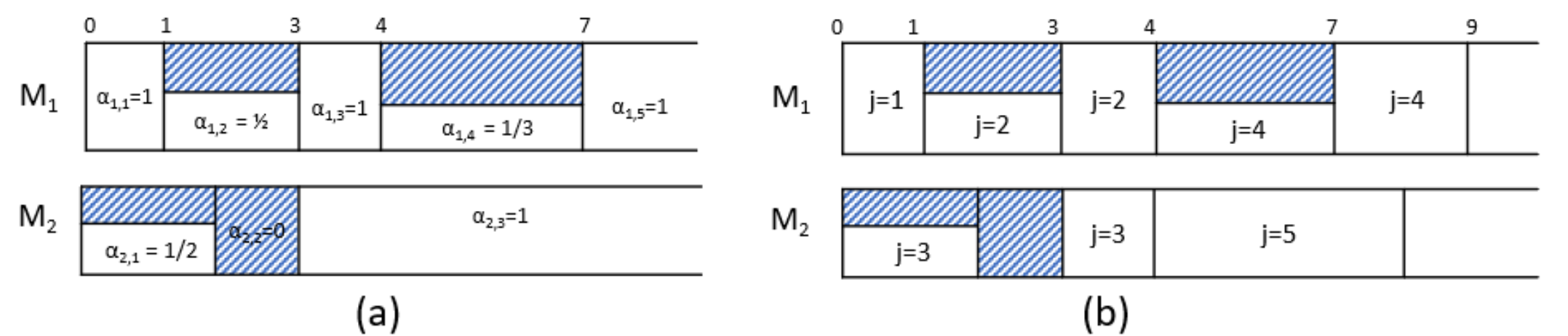}
\caption{ (a) Two machines with varying processing capacity, (b) A schedule of the 5  jobs, where $p_1 = 1$, $p_2 = p_3 = 2$, $p_4 = 3$, $p_5= 4$. } \label{fig:varying-processing-capacity}
\end{figure}
If the machine $M_i$ has full availability during an interval, we say the machine's processing capacity is $1$ and a job, $j$, can be completed in $p_j$ time units; otherwise, if the machine $M_i$'s processing capacity is less than $1$, for example $\alpha$, then  $j$ can be completed in $p_j/\alpha$ time units. 

The objective is to minimize the total completion time of all jobs. For any schedule $S$, let $C_j(S)$ be the completion time of the job $j$ in $S$. If the context is clear, we use $C_j$ for short. The total completion time of the schedule $S$ is $\sum_{1\le j \le n} C_j$. 

In this paper, we study our scheduling problem under data stream model. %Specifically, we study the data stream model of our problem where the number of jobs is so big that jobs' information cannot be stored but 
%and thus can only be scanned in one or more passes.
The goal is to design streaming algorithms to approximate the optimal solution in a limited number of passes over the data and using limited space.
%Formally, streaming algorithms are algorithms  in which the input is presented as a sequence of items and can be examined in only a few passes (typically just one). The performance of an algorithm that operates on data streams is measured by three basic factors: the number of passes the algorithm must run over the data stream, the space needed and the updating time of the algorithm.
Using the three-field %$\alpha \mid \beta \mid \gamma$
notation introduced by Graham et al. \cite{gllr79}, our problem is denoted as $P_m, \alpha_{i,k} \mid stream  \mid \sum C_j$; %if the processing capacity are arbitrary for all intervals on all machines; % for all shared intervals (routine jobs);
if for all intervals, the processing capacity is greater than or equal to a constant   $\alpha_0$, $0 < \alpha_0 \le 1$, our problem is denoted as $P_m, \alpha_{i,k} \ge \alpha_0  \mid stream \mid \sum C_j$.
%; and if the processing capacity is at least $\alpha_0$ for the intervals on the first $m_1$ ($ 1 \le m_1 \le m-1$) machines, % but arbitrary for other $(m-m_1)$ machines, our problem is denoted as $P_m, \alpha_{i \le m_1,k} \ge \alpha_0  \mid stream \mid \sum C_j$.

\subsection{Literature Review}
%The majority of the related literature assumes conventional data model. 
We first review the related work under the traditional data model.
The first related work is done by Baker and Nuttle \cite{bn80} in 1980. They studied the problems of sequencing $n$ jobs for processing by a single resource to minimize a function of job completion times subject to the  constraint that the availability of the resource varies over time. Their work showed that a number of well-known results for classical single-machine problems can be applied with little or no modification to the corresponding variable-resource problems.
Adiri and Yehudai \cite{ay87} studied the problem on single and parallel machines with the restrictions such that if a job is being processed, the service rate of a machine remains constant and the service rate can be changed only when the job is completed. In 1992, Hirayama and Kijima \cite{hk92} studied this problem on single machine where the machine capacity varies stochastically over time.

In 2016, Hall et. al. in \cite{hll16} studied similar problems in multitasking environment, where a machine does not always have the full capacity for processing primary jobs due to previously  scheduled routine jobs. In their work, they assume there is a single machine whose processing capacity is either 1 or a constant $e$ during an interval. They showed that the total completion time can be minimized by scheduling the jobs in non-decreasing order of the processing time, but  it  is unary NP-Hard for the objective function of total weighted completion time. %When the primary jobs have different due dates, the authors gave  polynomial time algorithms for  maximum lateness and the number of late jobs.

Another widely studied model is scheduling subject to machine unavailability constraint where
%. Under this model, machines may have some unavailable intervals due to breakdown or preventive maintenance. Here
the machine has either full capacity, or zero capacity so no  jobs can be processed. Various performance criteria and machine environments have been studied under this model, see the survey papers \cite{sc2000} and \cite{mcz10} and the references therein. Among the results, for the objective of minimizing total completion time on parallel machines, the problem is NP-hard.

Other scheduling models with  varying  processing capacity  are also studied in the literature where variation of machine availability are related  with the jobs that have been scheduled. These models include scheduling with learning effects (see the survey paper \cite{jkt11} by Janiak et al. and the references therein), scheduling with deteriorating effects (see the paper \cite{ccc10} by Cheng et al.),  and  interdependent processing capacitys (see the paper \cite{awkl19} by Alidaee et al. and the references therein), etc.

In our model, there are multiple machines, the processing capacity of each machine can change between 0 and 1 from interval to interval and the capacity change can happen while a job is being processed. The goal is to find a schedule to minimize total completion time. 
In \cite{fhz22} we %have studied this problem under the traditional data model.% where all data can be stored locally and thus easily accessed by our classical algorithms in constant time. 
have showed that there is no polynomial time approximation algorithm unless $P=NP$ if the processing capacities  %, which are the fractions of the processing capacity assigned to primary jobs, 
are arbitrary for all machines. Then for the problem where the sharing ratios on some machines have a constant lower bound, we analyzed the performance of some classical scheduling algorithms and developed a polynomial time approximation scheme
%that runs in linear time 
when the number of machines is a constant.

%So far there are no results about the problems studied in this paper. Note that  if the processing capacity is always 1, that is, all machines have full availability, our problem becomes the classical parallel machine scheduling problem $P_m \mid \mid \sum C_j$ which can be solved using SPT (Shortest Processing Time First) rule. On the other hand, if the processing capacity is either 0 or 1 at any time, i.e.the machine either has full availability or is totally not available for processing the jobs, then our problem reduces to the problem of parallel machine scheduling with unavailability constraint where jobs can be resumed on after being interrupted. %When the jobs are non-resumable, The problem is NP-hard if the jobs are not allowed to migrate from one machine to another. 

%Streaming algorithms were studied by Munro and Paterson~\cite{mp78}, as well as Flajolet and Martin~\cite{fm85}. 
We now review the related work under the data stream model. Since the streaming algorithm model was formally established by
Alon, Matias, and Szegedy in ~\cite{ams99}, it has received a lot of attention.
While streaming algorithms have been studied in the field of statistics, optimization, and graph algorithms (see surveys by Muthukrishnan \cite{m05} and McGregor \cite{mg14}) and some research results have been obtained, very limited research on streaming algorithms (\cite{cv21} \cite{fhz22-2}) has been done  in the field of sequencing and scheduling so far. For the proposed model, in \cite{fhz22-2} we developed a streaming algorithm for the problem with the objective of makespan minimization. %There is no streaming algorithm in the literature for our studied problems.

\subsection{New Contribution}

In this paper, 
we present the first efficient streaming  algorithm  for the problem
$P_m, \alpha_{i,k}  \ge \alpha_0 \mid stream \mid \sum C_j$ . Our streaming algorithm can compute an
approximation of the optimal value of the total completion time in
one pass, and  output a schedule that approximates the optimal
schedule in two passes.

\section{A PTAS for $P_m, \alpha_{i,j} \ge \alpha_0 \mid stream \mid \sum C_j$}
In this section, we develop a streaming algorithm  for our problem when the processing capacities on all machines have a positive constant lower bound, that is, $P_m, \alpha_{i,k} \ge \alpha_0\mid stream \mid \sum C_j$. The algorithm is a PTAS and uses sub-linear space when $m$ is a constant.

At the conceptual level, our algorithm has the following two stages:
\begin{itemize}
    \item[] \textbf{Stage 1:} Generate a sketch of the jobs while reading the input stream.  
    \item[] \textbf{Stage 2:} Compute an approximation of the optimal value based on the sketch.  
\end{itemize}

In the following two subsections, we will describe each stage in detail. 

\subsection{Stage 1: Sketch Generation from the Input Stream}
In this step we read the job input stream $N = \{i: 1 \le i \le n\}$ and generate a sketch of $N$ using rounding technique. The idea of the procedure in this stage is to split the jobs into large jobs and small jobs and round the processing time of the large jobs so that the number of distinct processing times is reduced. Specifically, the sketch is a set of pairs where each pair contains a rounded processing time and the number of jobs with this rounded processing time. %specifies the number of large jobs for each distinct processing time after rounding. 
Let  $p_{max} = \max_{j \in N} {p_j}$.  The {\bf sketch}  of $N$ can  be formally defined as follows:
\begin{definition}
Given the error parameter $\epsilon$ and the lower bound of processing capacity of the machines $\alpha_0$, let $N_L = \{ j \in N:  p_j  \ge  \tfrac{\epsilon \alpha_0}{3n^2} p_{max}\} $ denote the set of large jobs.  Let $\tau =  \tfrac{ \epsilon \alpha_0}{15}$, for each job $j$ of $N_L$, we round up its processing time such that if  $p_j \in [ (1+\tau)^{k-1}, (1  + \tau)^{k}) $, then its rounded processing time is $rp_k=\floor{(1+\tau)^{k}}$. We denote the {\bf sketch}  of $N$ by  $N'_L = \{(rp_k, n_k):  k_0 \le k \le k_1 \}$, where $n_k$ is the number of jobs whose rounded processing time is $rp_k$ and $k_0$ and $k_1$ are the integers such that $\tfrac{\epsilon \alpha_0}{3n^2} p_{max} \in [(1+\tau)^{k_0-1}, (1+\tau)^{k_0})$ and $p_{max} \in  [(1+\tau)^{k_1-1}, (1+\tau)^{k_1})$ respectively. 
\end{definition}

%From the above definition, we can see that there are at most $k_1-k_0+1$ distinct processing times in $N'_L$. 
When $n$ and $p_{max}$ are known before reading the job stream, one can generate the sketch from the job input stream  with the following simple procedure: Whenever a job $j$ is read, if it is a small job, skip it and continue. Otherwise, it is a large job. Suppose that $p_j \in  [(1+\tau)^{k-1}, (1  + \tau)^{k})$, we update  the pair $(rp_k, n_k)$  by increasing $n_k$ by 1, where $rp_k = \floor{(1+\tau)^{k}}$. 

In reality, however, $n$ and $p_{max}$ may not be obtained accurately without scanning all the jobs. Meanwhile, in many practical scenarios, the estimation of $n$ and $p_{max}$ could be obtained based on priori  knowledge. Specifically, an upper bound of $n$, $n'$, could be given such that $n \le n' \le c_1n$ for some $c_1 \ge 1$; and a lower bound of $p_{max}$, $p'_{max}$, could be given such that $1\le p_{max}'\le p_{max}\le  c_2p_{max}'$ for some $c_2 \ge 1$.
Depending on whether $n'$ and $p'_{max}$ are known beforehand, we have four cases: (1) both $n'$ and $p'_{max}$ are known; (2) only $p'_{max}$ is known; (3) only $n'$ is known; (4) neither $n'$ nor $p'_{max}$  is known. 

For all four cases, we can follow the same procedure below to get the sketch of the job input stream. The main idea is as we read each job $j$, we dynamically update the maximum processing time, $p_{curMax}$, the total number of jobs, $n_{cur}$,  the threshold of processing time for large jobs, $p_{minL}$,  and  the sketch if needed.   For convenience, in the following procedure we treat $\infty$ as a number, and $1/\infty$ as  $0$.

\medskip
{\noindent \bf Algorithm for constructing sketch $N'_L$}   

 \begin{enumerate} [label={\arabic*.}]
    \item  Let $\tau =  \tfrac{ \epsilon \alpha_0}{15}$
    \item if $n'$ is not given, set $n' = \infty$
    \item if $p'_{max}$ is not given, set  $p'_{max} = 1$ 
    \item $p_{minL} = \max\{\tfrac{\epsilon\alpha_0}{3(n')^2} p'_{max}, 1 \}$,
     \item Initialize 
     $p_{curMax}= 0$, $n_{cur} = 0$, and  $N''_{L} = \emptyset$ 
     
    \item Construct $N''_{L}$ while repeatedly reading next $p_j$
    \begin{enumerate}[label*={\alph*.}]
        \item $n_{cur} = n_{cur}+1$
       \item  If $p_j > p_{curMax}$
         \item[]   \hspace{0.2in}     $p_{curMax}=p_j$
       \item[] \hspace{0.2in} if $p_{minL}  < \tfrac{\epsilon\alpha_0} { 3(n')^2} \cdot p_{curMax}$, $p_{minL}  = \tfrac{\epsilon\alpha_0} { 3(n')^2} \cdot p_{curMax}$

       \item If $p_j \ge   p_{minL}$ 
             \begin{enumerate} [label*={\arabic*}]
             %\item $ rp =  (1 + \tau)^{\ceiling{\log_{1+\tau} p_j}}$
             \item $rp = \floor{(1+\tau)^k}$ where $p_j \in [ (1+\tau)^{k-1}, (1  + \tau)^{k}) $
             \item if there is a tuple $(rp_k, n_k)  \in N''_L$ where $rp_k = rp$, 

             \item  then update $n_k = n_k +1$
              \item \label{sketch-insert} else 
                     
                     \item[] \hspace{0.2in} $N''_{L} = N''_{L} \cup \{( rp, 1) \}$ 
              \end{enumerate} 
    \end{enumerate}
    \item $p_{max} = p_{curMax}$, $n = n_{cur}$,  $p_{minL} = \tfrac{\epsilon\alpha_0}{3n^2} p_{max}$
    \item Let $N'_L = \{ ( rp_k, n_k): ( rp_k, n_k) \in N''_L, \text { and } rp_k > p_{minL} \}$
    
    \end{enumerate}

While the above procedure can be used for all four cases, we need  different data structures and implementations  in each case to achieve time and space efficiency. For cases (1) and (2) where $p'_{max}$ is known, since $1\le p_{max}'\le p_{max}\le  c_2p_{max}'$ for some constant $c_2$, there are at most $\log_{1+ \tau}  { c_2 p'_{max} }$ distinct rounded processing times, and thus we can use an array to store $N''_L$. For cases (3) and (4) where no information about $p'_{max}$ is known, we can use a B-tree  to store the elements of   $N''_L$. Each node in the tree corresponds to an element $(rp_k, n_k)$ with $rp_k$ as the key. With $p_{curMax}$ being dynamically updated, there are at most $\log_{1+ \tau}  { p_{curMax} }$ distinct rounded processing times, and thus at most $\log_{1+ \tau}  {p_{curMax} }$ nodes in the B-tree at any time.  %$n_{cur}$, and $p_{minL}$ as each job $j$ is read in, we only store in B-tree the pairs $(rp_k, n_k)$ such that $p_{minL} \le rp_k \le p_{curMax}$ and $n_k \ge 1$. 
As each job $j$ is read in, we may need to insert a new node to B-tree. If $p_j>p_{curMax}$, $p_{curMax}$ needs to be updated and so does $p_{minL}$, which is the threshold of processing time for large jobs. Hence the nodes with the key less than $p_{minL}$ should be removed. To minimize the worst case update time for each job, each time when a new node is inserted, we would delete the node with the smallest key if it is less than $p_{minL}$. % as well as delete the nodes that has $rp_k<p_{minL}$. 

The  following lemma gives   the space and time complexity for computing sketch $N'_L$ from the job input stream for all four cases.

\begin{lemma}\label{lemma:computing-sketch} Let  $\alpha_0$ and $\epsilon$ be real numbers in $(0,1]$. We can compute the sketch $N'_L$ of job input stream in one pass with the following performance:
    \begin{enumerate}
        \item Given both an upper bound $n'$ for $n$  and a lower bound $p'_{max}$ for $p_{max}$ such that   $n\le n'\le c_1 n$ and   $1\le p_{max}'\le p_{max}\le  c_2p_{max}'$ for some $c_1$ and $c_2$,
        then it takes $O(1)$  update time and
         $O(\tfrac{1}{ \epsilon\alpha_0}\min (\log n+\log \tfrac {c_1c_2} {\epsilon\alpha_0},\log p_{max}+\log c_2))$ space to process each job from the stream. 
      \item Given only a lower bound $p'_{max}$ for $p_{max}$ where  $1\le p_{max}'\le p_{max}\le  c_2p_{max}'$,  then it  takes  $O(1)$  update time  and 
         $O(\tfrac{1}{\epsilon\alpha_0}(\log p_{max}+\log c_2))$ space to process each job in the stream.
    \item Given only  an upper bound $n'$ for $n$ such that   $n\le n'\le c_1 n$,  then   it  takes  $O(\log (\tfrac{1}{ \epsilon\alpha_0}) + \min( \log (\log n+\log \tfrac{c_1}{  \epsilon\alpha_0})), \log \log p_{max})$  update time, and 
         $O(\tfrac{1}{\epsilon\alpha_0}\min(\log n+\log \tfrac{c_1}{  \epsilon\alpha_0},\log p_{max}))$ space to process each job in the stream.
     \item Given no information about $n$ and $p_{max}$, then it takes  $O(\log (\tfrac{1}{ \epsilon\alpha_0} + \log \log p_{max}))$  update time, and      $O(\tfrac{1}{ \epsilon\alpha_0}\log p_{max})$ space to process each job in the stream.
    \end{enumerate}
\end{lemma}
\begin{proof} 
We give the proof for four cases separately as follows:

{\noindent \bf Case 1:}  Both $n'$ and $p'_{max}$ are given such that $n\le n'\le c_1 n$, and $1\le p_{max}'\le p_{max}\le  c_2p_{max}'$ for some $c_1 >1$ and $c_2 > 1$.

\noindent From the algorithm, the processing time of a large job is at most $c_2 p'_{max}$   and  at least $p_{minL}=\max\{\tfrac{\epsilon\alpha_0}{3(n')^2}p'_{max}, 1\}$. Thus, the number of distinct processing times after rounding is at most $n''= \min (\log_{1+ \tau} \tfrac {c_2 \cdot 3(n')^2} {\epsilon\alpha_0}, \log_{1+ \tau}  { c_2 p'_{max} })$.  
We use an array of size $n''$ to store the elements of $N''_L$ and we have 
 \begin{align*}
n'' &= \min(\log_{1+ \tau} \tfrac{c_2 \cdot 3 (n') ^2} { \epsilon\alpha_0},  \log_{1+ \tau} c_2 p'_{max} ) \\  
&  =    \log_{1+ \tau} \min (  \tfrac{c_2 \cdot 3 (n') ^2} { \epsilon\alpha_0}, c_2 p_{max}) \\
 &  \le   \log_{1+ \tau} \min ( \tfrac{c_2 \cdot 3 c_1^2 n ^2} { \epsilon\alpha_0},  c_2 p_{max}) \\
  & = O(\tfrac{1}{ \tau}\min (\log n+\log \tfrac {c_1c_2} {\epsilon\alpha_0},\log p_{max}+\log c_2))\\
 & = O(\tfrac{1}{ \epsilon\alpha_0}\min (\log n+\log \tfrac {c_1c_2} {\epsilon\alpha_0},\log p_{max}+\log c_2)). 
\end{align*} 

It is easy to see  that  the update time for each job is $O(1)$ time.    
 
\smallskip
 {\noindent \bf Case 2:}
Only  $p'_{max}$,  $p_{max}'\le p_{max}\le  c_2p_{max}'$, is given. 

\noindent From the algorithm, the processing time of a large job is between $p_{minL} = 1$ and   $c_2 p'_{max}$,  thus the number of distinct processing time in $N''_L$ is at most $n''=   \floor{\log_{1+\tau} c_2 p'_{max}} \le \floor{\log_{1+\tau} c_2 p_{max}}  = O(\tfrac{1}{ \epsilon\alpha_0}(\log p_{max}+\log c_2))$. 
 
With the array of $n''$ elements to store the elements of  $N''_L$, the update time for each job is $O(1)$.
  
\smallskip
{\noindent \bf Case 3:} Only  $n'$, $n \le n' \le c_1 n$, is given.
 
We use a B-tree to store the elements of $N''_L$. Since $n'$ is given, we can calculate $p_{minL}$ based on the updated $p_{curMax}$, $p_{minL}=\max\{\tfrac{\epsilon\alpha_0}{3(n')^2}p_{curMax}, 1\}$.
%Let $rp_{min} = \min_{(rp_k, n_k) \in N''_L}{rp_k}$. Since we update $p_{minL}$ dynamically, it is possible that $rp_{min} < p_{minL}$. To control the size of B-tree, whenever we insert a pair $(rp_k, n_k)$ into the tree, we do the following: if $rp_{min}  < p_{minL}$, delete the node with key $rp_{min}$.  In this way, 
So the number of nodes in the B-tree   is bounded by $\tfrac{p_{curMax}}{p_{minL}}$, which is  
\begin{align*}
 \min ( \log_{1+\tau} \tfrac {3n'^2}{\epsilon\alpha_0}, \log_{1+\tau} p_{max})) 
& = \log_{1+\tau}(\min ( \tfrac {3n'^2}{\epsilon\alpha_0}, p_{max})) \\
& \le \log_{1+\tau}(\min ( \tfrac{3 c_1^2 n^2}{  \epsilon\alpha_0},  p_{max})) \\
& =O(\tfrac{1}{\tau}\min (\log n+\log {c_1\over \epsilon\alpha_0},\log p_{max})) \\
& = O(\tfrac {1}{ \epsilon\alpha_0}\min(\log
n+\log \tfrac{c_1}{  \epsilon\alpha_0},\log p_{max})).
\end{align*}

For each large job, we need to perform at most three operations: a search operation, possibly  an insertion and a deletion. The time for each operation is at most the height of the tree:  
\begin{eqnarray*}  & & \log (O(\tfrac {1}{ \epsilon\alpha_0}\min(\log
n+\log \tfrac{c_1}{  \epsilon\alpha_0},\log p_{max}))) \\
& =  &  O(\log (\tfrac{1}{ \epsilon\alpha_0}) + \min( \log (\log n+\log \tfrac{c_1}{  \epsilon\alpha_0}), \log \log p_{max}))
\end{eqnarray*}
\smallskip
{\noindent \bf Case 4:} No information about $p_{max}$ and $n$ is known beforehand. We still use B-tree as Case 3. However, without information of $n$, $p_{minL}$ is always 1, the total number of nodes stored in the B-tree is at most $  O(\log_{1+\tau} p_{max})=O(\tfrac{1}{ \epsilon\alpha_0}\log
p_{max})$. The update time is thus $ O(\log (\log_{1+\tau} p_{max}) ) = O(\log \tfrac{1}{ 
\epsilon \alpha_0}+\log \log p_{max})$.
\end{proof}

\subsection{Stage 2: Approximation Computation based on the Sketch} 

In this stage, we will find an approximate value of the minimum total completion time for our scheduling problem based on the sketch $N'_L=\{(rp_k, n_k): k_0\le k \le k_1\}$ that is obtained from the first stage. The idea is to assign large jobs in the sketch $N'_L$ group by group in SPT order to $m$ machines where group $k$, $k_0 \le k \le k_1$, corresponds to pair $(rp_k, n_k)$ in sketch $N'_L$. After all groups of jobs are scheduled, we find the minimum total completion time among  the remaining  schedules, and return an approximate value. 

To schedule each group of jobs, we perform two operations:
\begin{itemize}[nosep]
\item[]{\bf  Enumerate: } enumerate all assignments of the jobs in the group to $m$ machines that satisfy certain property, 
and 
\item[]{\bf  Prune:} prune the schedules so that only a limited number of schedules are kept.
\end{itemize}
 %no two schedules are  ``similar'',  where ``similar'' schedules are defined as follows.

During the Enumerate operation, we enumerate the assignments of $n_k$ jobs from group $k$ to $m$ machines using {\it $(\delta, m)$-partition} as defined below.

\begin{definition}
For two positive integers $b$ and $m$, and a real $\delta>0$, a {\it $(\delta, m)$-partition} of $b$ is an ordered tuple $(b_1,b_2,\cdots, b_m)$
such that each $b_i$ is non-negative integer, $b= \sum_{1 \le i \le m} b_i$,
and there are at least $(m-1)$ integers $b_i$ is either  $0$ or
$\floor{(1+\delta)^q}$ for some integer $q$.
\end{definition}

For example, for $\delta=1$, to schedule a group of  $b = 9$ jobs to $m = 3$ machines, we enumerate those assignments corresponding to $(1, 3)$-partitions of $9$. From the  definition, some examples of $(1, 3)$-partitions of 9 are   $\{2, 2, 5\}$,  $\{0, 9, 0\}$,  $\{0, 1, 8\}$. Corresponding to the partition $\{2, 2, 5\}$, we schedule 2 jobs on the first machine, 2 jobs on the second machine and the remaining 5 jobs on the last machine. By definition,   $\{2, 2, 5\}$ and $\{2, 5,  2\}$ are two different partitions corresponding to two different schedules.

During the Prune operation, we remove some schedules so that only a limited number of schedules are kept. 
Let $S$ be a schedule of the jobs on $m$ machines, we use   $P_i(S)$ to denote the total processing time of the jobs assigned to machine $M_i$ in $S$; and use $\sigma_i(S)$ to denote  the total completion time of the jobs scheduled to $M_i$ in $S$. The schedules are pruned so that no two schedules are  ``similar'',  where ``similar'' schedules are defined as follows.

 \begin{definition}
 Two schedules  $S_1$, $S_2$  are ``similar'' with respect to a given
parameter $\delta$  if for every $1 \le i \le m$, $P_i(S_1)$ and
$P_i(S_2)$ are both in an interval $[(1+\delta)^x, (1 +
\delta)^{x+1})$ for some integer $x$, and $\sigma_i(S_1)$ and
$\sigma_i(S_2)$ are both in an interval $[(1+\delta)^y, (1 +
\delta)^{y+1})$ for some integer $y$. We use $S_1 \overset{\delta}{
\approx}  S_2$ to denote $S_1$ and $S_2$ are ``similar'' with
respect to $\delta$. 
 \end{definition}

%For convenience, we refer the jobs with same processing time $rp_k$ as the jobs in the $k$-th group. Our algorithm schedules the jobs from $N'_L$ group by group in increasing order of the processing time. To reduce the time complexity, we only consider those schedules so that the number of jobs assigned to $m$ machines  form a $(\delta, m)$-partition of $n_k$, where $\delta$ is a parameter to be determined later.  To further reduce the time and space complexity, we prune the  schedules so that no two schedules are  ``similar''. The complete algorithm is given below.

Our complete algorithm for stage 2 is formally presented below in which  the Enumerate and Prune operations are performed in Step (\ref{enumerate-step}) and   (\ref{prune-step}), respectively.
  
\medskip
{\noindent \bf Algorithm for computing the approximate value}  

\noindent {\bf Input}: $N'_L = \{ (rp_k, n_k):  k_0  \le k \le k_1 \}$
 
\noindent{\bf Output:} An approximate value of the minimum total completion time for jobs in $N$
\smallskip

\noindent{\bf Steps}
\begin{enumerate}[label={\arabic*.}]
\item let $\delta $ be a positive real such that $\delta  < \tfrac{\epsilon \cdot \alpha_0}{24 (k_1 - k_0 + 1)}$
\item $U_{k_0-1} = \emptyset$
\item Compute $U_k$, $  k_0 \le k \le k_1$, which is a set of schedules of jobs from the groups $k_0$ to $k$  
\begin{enumerate} [label*={\alph*}]
\item let $U_{ k} = \emptyset$
\item \label{enumerate-step} for each $(\delta, m)$-partition  of $n_k$

    \begin{enumerate}
     \item[] for each schedule $S_{k-1} \in U_{ k-1}$
     \item[] \hspace{0.15in} 
      schedule the jobs of the group $k$ to the end of $S_{k-1}$ based
       \item[] \hspace{0.15in} on the partition and let the  new schedule be $S_k$
      %\item[] \hspace{0.15in}  let the  new schedule be $S_k$
     \item[] \hspace{0.15in}  $U_{k} = U_{ k} \cup \{S_k\}$
    \end{enumerate}
\item \label{prune-step}prune $U_k$ by repeating the following until $U_k$ can't be reduced
    \begin{enumerate}
    \item [] if there are two schedules $S_1$ and $S_2$ in $U_k$ such that $S_1 \overset{\delta }\approx S_2$
    \item[] \hspace{0.2in} $U_{ k} = U_{ k} \setminus \{S_2\}$
    \end{enumerate}
\end{enumerate}
\item Let  $S' \in U_{k_1}$  be the schedule  that has the minimum   total completion time, $\sigma(S')= \Sigma _{1 \le i \le m} \sigma_i(S')$. 
\item return $(1 + {\epsilon}/{3}) (1+{\epsilon}/{15})\sigma(S') $  as an approximate value of the minimum total completion time of the jobs in $N$.
\end{enumerate}

Before we analyze the performance of the above procedure, we first consider a special case of our scheduling problem: all jobs have equal processing time and there is a single machine that has the processing capacity at least $\alpha_0$ at any time.  Suppose that the jobs are continuously scheduled on the machine. The following lemma shows how the total completion time of these jobs changes if we shift the starting time of these jobs and/or insert additionally a small number of identical jobs at the end of these jobs.

\begin{lemma}\label{claim-same-length-job}
Let $S_x$ be a schedule of $x$ identical jobs of processing time $p$
starting from time $t_0$ on a single machine whose processing capacity is
at least $\alpha_0$ at any time.  Then we have the following cases:
 \begin{enumerate}
 [label={(\arabic*)}]
    \item\label{prune-claim-1}   $  x \cdot t_0 + \tfrac{x(1+x)}{2} p   \le  \sigma (S_x) \le  x \cdot t_0 +  \tfrac{x(1+x)}{2 \cdot \alpha_0} p $.

     \item\label{prune-claim-2} If we shift all jobs in $S_x$ so the first job starts at $(1+\delta)t_0$ and get a new schedule $S_x^{1}$, then   $ \sigma(S_x^{1}) \le  ( 1 + \tfrac{\delta}{\alpha_0}) \sigma (S_x) $.

     \item\label{prune-claim-3} If we add additional $ \floor{x \delta}  $ identical  jobs at the end of $S_x$ and  get a new schedule $S_x^{2}$, then its total completion time is $ \sigma(S_x^{2}) \le  ( 1 + \tfrac{3\delta}{\alpha_0}) \sigma(S_x)$.

     \item\label{prune-claim-4} Let $S_x^{3}$ be a schedule  of $(x   + \floor {x \delta'}) $ identical jobs of processing time $p$ starting from time $(1 + \delta'')t_0$
      then $\sigma(S_x^{3}) \le  ( 1 + \tfrac { \delta'' } {\alpha_0} ) ( 1 + \tfrac {3 \delta' } {\alpha_0} ) \sigma(S_x)$
 \end{enumerate}
\end{lemma}

\begin{proof}
We will prove \ref{prune-claim-1}-\ref{prune-claim-4} one by one in order.
\begin{enumerate}
[label={(\arabic*)}]
      \item First it is easy to see that $x\cdot t_0 +\tfrac{x(1+x)}{2} p $ is the total completion time of the jobs when the machine's processing capacity is always 1, which is obviously a lower bound of $\sigma(S_x)$. When the machine's processing capacity is at least $\alpha_0$, then it takes at most $p/\alpha_0$ time to complete each job, thus   the total completion time is at most $  x \cdot t_0 + \tfrac{x(1+x)}{2 \cdot \alpha_0} p $.

     \item When we shift the jobs so that the first job starts $\delta t_0$ later, then the completion time of each job is increased by at most $\tfrac{ \delta t_0}{\alpha_0}$. Therefore,
     $$\sigma(S_x^{1}) \le \sigma(S_x) +  x \cdot \tfrac{ \delta   \cdot t_0 }{\alpha_0} \le  (1 + \tfrac{ \delta}{\alpha_0}) \sigma(S_x) \enspace.$$

     \item Suppose the last job in $S_x$ completes at time $t$. Then $t   \le t_0 + \tfrac{x p }{\alpha_0} $. When we add additional $\floor{x \delta}  $ jobs starting from $t$,  by \ref{prune-claim-1}, the total completion time of the additional jobs is at most $ (x \delta \cdot t  +  \tfrac{ x \delta (1+x\delta)}{2 \cdot \alpha_0} p ) $. Therefore,
     \begin{eqnarray*}
      \sigma(S_x^{2}) &  \le &
      \sigma(S_x) +  x \delta \cdot t  + \tfrac{ x \delta (1+x\delta)}{2 \cdot \alpha_0} p   \\
      &  \le &
        \sigma(S_x) +       x \delta \cdot  (t_0 + \tfrac{x p }{\alpha_0}) + \tfrac{\delta}{ \alpha_0} \tfrac{x ( 1 + x \delta)}{2} p \\
      &  \le &
        \sigma(S_x) +   \tfrac{\delta}{ \alpha_0}  (   x  \cdot t_0 + \tfrac{x ( 1 + x )}{2} p ) +      x \delta \cdot   \tfrac{x p }{\alpha_0} \\
     &  \le &  \sigma(S_x) +   \tfrac{\delta}{ \alpha_0}  \sigma(S_x) +   \tfrac{ 2 \delta}{ \alpha_0}  \sigma(S_x) \\
     &  \le &  ( 1 + \tfrac {3 \delta } {\alpha_0} )\sigma(S_x) \enspace.
     \end{eqnarray*}

     \item Follows from \ref{prune-claim-2} and \ref{prune-claim-3}.
 \end{enumerate}
\end{proof}
 
Now we analyze the performance of our algorithm. In step \ref{enumerate-step},  we only consider the schedules of the $n_k$ jobs corresponding to  {\it $(\delta, m)$-partitions}. Let $S$ be any schedule of the jobs in sketch $N'_L$, we will show that at the end of step \ref{enumerate-step}  there  is a schedule   $S_{\delta} \in U_{k}$  that is $\delta_k$-close  to $S$. Let $n_{i, k}(S)$ be the number of jobs from group $k$ that are scheduled on machine $M_i$ in $S$. The $\delta_k$-close schedule to $S$ is defined as follows. 

\begin{definition}
Let $k$ be an integer such that $k_0 \le k \le k_1$. 
We  say a schedule 
$S_{\delta}$  is a  {\bf $\delta_k$-close schedule to $S$} if
for the jobs in group $k$, the following conditions hold:
(1) In $S_{\delta}$, the schedule of jobs from the group $k$
form  a $(\delta, m)$ partition  of $n_k$; 
(2)  For at least $m-1$ machines, either $n_{i, k}(S_{\delta}) = 0$ or $ \ceiling{\log_{1 + \delta}n_{i, k}(S_{\delta})} = \ceiling{\log_{1 + \delta}n_{i, k}(S)}$.
%If $S_{\delta}$  is a  $\delta_k$-close schedule to $S$ for all $k_0 \le k \le k_1$, we say $S_{\delta}$  is a  {\bf $\delta$-close schedule to $S$}.
\end{definition}  

\noindent By definition, if $S_{\delta}$ is $\delta_k$-close to $S$, then  $n_{i,k}(S_{\delta}) \le (1 + \delta) n_{i,k}(S)$ for all $i$,  $1 \le i \le m$. 

The following lemma shows that  
there is always a schedule $S_{\delta} \in U_k$ at the end of step \ref{enumerate-step} that is $\delta_k$-close to $S$.
%, and its total completion time is close to that of $S$.

\begin{lemma}\label{claim-delta-close-sumcj} For any schedule $S$, 
there exists a $\delta_k$-close schedule $S_{\delta} \in U_k$ at the end of step \ref{enumerate-step} that is $\delta_k$-close to $S$.
%and $\sigma(S_{\delta})\le (1+ \tfrac{3\delta}{\alpha_0})(1+\tfrac{\delta}{ \alpha_0}) \cdot \sigma(S)$. 
\end{lemma}
\begin{proof}
%
%We assume $rp_1< rp_2<\cdots < rp_{\mu}$.
The existence of $S_\delta$ can be shown by construction. %Initially, $S_{\delta}$ is empty. 
%We schedule the jobs in groups starting from group $k=k_0$ until group $k_1$, which are in SPT order. 
We initialize $S_{\delta}$ to be any schedule from $U_{k-1}$. Then we schedule  the jobs from group $k$ to $M_i$, starting from $i=1$.  Suppose there are $n_{i,k}(S) > 0 $ jobs scheduled on $M_i$ in $S$, and  $  (1 +\delta)^{q-1} < n_{i, k}(S) \le (1 +\delta)^q$ for some integer $q \ge 0$,   then if there are less than $\floor{ (1 +\delta)^q}$ jobs unscheduled in this group, assign all the remaining jobs to machine $M_i$; otherwise, assign $\floor{ (1 +\delta)^q}$ jobs to machine $M_i$ and continue to schedule the jobs of this group to the next machine.

It is easy to see that the constructed schedule of jobs from group $k$ form a $(\delta, m)$-partition that would be  added to $U_k$ in step \ref{enumerate-step}. By definition, the constructed schedule $S_\delta$ is  a {\bf $\delta_k$-close schedule to $S$} . %if schedules are not pruned. % although ``similar'' schedules may be pruned in each iteration after group $k$ of jobs are scheduled.
%And we have $n_{i,k}(S_{\delta}) \le (1 + \delta) n_{i,k}(S)$ for all $i$,  $1 \le i \le m$, and $k$, $k_0 \le k \le k_1$. It is easy to see that for each machine $M_i$ and each $k$ the total processing time of the jobs from the first $k-k_0+1$ groups in $S_{\delta}$ is no more than $1+\delta$ times the one in $S$, and thus we have $\mathcal{T}_{i, k}(S_\delta) \le ( 1 +\tfrac{\delta}{\alpha_0})\mathcal{T}_{i, k}(S)$. And by Lemma~\ref{claim-same-length-job}, we have $\sigma_{i, k}(S_{\delta}) \le ( 1 + \tfrac{3 \delta}{\alpha_0}) ( 1 + \tfrac{\delta}{\alpha_0})\sigma_{i, k}(S)$, which implies that $\sigma(S_{\delta})\le (1+ \tfrac{3\delta}{ \alpha_0})(1+\tfrac{\delta}{ \alpha_0}) \cdot \sigma(S)$.
 \end{proof}

We now analyze step \ref{prune-step} of our algorithm where  ``similar'' schedules are pruned after a group of jobs in $N'_L$ are scheduled. 
We will need the following notations:

\rightskip=2em \leftskip=2em \noindent $\sigma_{i, k}(S)$:   the
total completion time of jobs from group $k$ that are
scheduled on $M_i$ in   $S$.

\noindent $\mathcal{T}_{i, k}(S)$: the largest completion time of the jobs from group $k$ that are scheduled on $M_i$ in $S$.

\noindent $P_{i, k}(S)$: the total processing time of jobs from group $k$ that are scheduled to machine $M_i$ in  $S$.

\rightskip=0em \leftskip=0em

For any optimal schedule  $S^*$ of the jobs in $N'_L$, let $S_k^*$ be the partial schedule of jobs from groups $k_0$ to $k$ with the processing time at most $rp_k$ in $S^*$. Our next lemma  shows that there is a schedule $S_k \in U_k$ that approximates the partial schedule $S_k^*$.

\begin{lemma}\label{lemma-close-approx-after-prune}
For any optimal schedule  $S^*$ of the jobs in $N'_L$, let $S_k^*$ be the
partial schedule in $S^*$ of the jobs from groups $k_0$ to $k$. Let $\mu = k_1 - k_0 +1 $, then after
some schedules in $U_k$ are pruned at step (\ref{prune-step}), there
exists a schedule $S_k \in U_k$  such that
\begin{enumerate}
    \item[(1)]   $P_i(S_k)  \le (1 + \delta)^{k-k_0+2} P_i(S^*_k)$ for  $1 \le i \le m$, and
    \item[(2)]
     $\sigma_i(S_k) \le (1 + \delta)^{k-k_0+1} (1+ \tfrac{ 2 \mu \delta}{\alpha_0})(1 + \tfrac{3\delta}{\alpha_0}) \sigma_i(S^*_k)$ for $1 \le i \le m$.
\end{enumerate}
\end{lemma}

\begin{proof}
We prove by induction on $k$. First consider $k=k_0$. By Lemma~\ref{claim-delta-close-sumcj}, at the end of  step \ref{enumerate-step}  there is a schedule $S^{\delta}_{k_0} \in U_{k_0}$ that is $\delta_{k_0}$-close to $S^*_{k_0}$, and $n_{i,k_0}(S^{\delta}_{k_0}) \le (1 + \delta) n_{i,k_0}(S^*_{k_0})$ for all $i$,  $1 \le i \le m$ which implies $p_{i,k_0}(S^{\delta}_{k_0}) \le (1 + \delta) p_{i,k_0}(S^*_{k_0})$. In both schedules $S^{\delta}_{k_0}$ and  $S^*_{k_0}$, the jobs are scheduled from time 0 on each machine, by
Lemma~\ref{claim-same-length-job} Case (3), for each machine $M_i$, we have
$\sigma_i(S^{\delta}_{k_0}) \le  ( 1 + \tfrac{3\delta}{\alpha_0})
\sigma_i(S^*_{k_0})$. If $S^{\delta}_{k_0}$ is pruned from $U_{k_0}$ at step \ref{prune-step}, then there must be   a schedule $S_{k_0}$ such that
$S_{k_0} \in U_{k_0}$ and $S^{\delta}_{k_0} \overset{\delta}\approx S_{k_0}$, so
for each machine $M_i$,  $1 \le i \le m $, we have 
$$P_i(S_{k_0}) \le (1 + \delta) P_i(S^{\delta}_{k_0}) \le (1 + \delta)^2 P_i(S^*_{k_0}) = (1 + \delta)^{k-k_0+2} P_i(S^*_{k_0})$$ and
$$\sigma_i(S_{k_0}) \le (1 + \delta)\sigma_i(S^{\delta}_{k_0}) \le
(1+\delta) (1+ \tfrac{ 3 \delta}{\alpha_0}) = (1+\delta)^{k-k_0+1} (1+ \tfrac{ 3 \delta}{\alpha_0}) \sigma_i(S^*_{k_0}).$$

Assume the induction hypothesis holds for some $k \ge k_0$. So after schedules in $U_k$ are pruned,   there is a schedule  $S_k$ in $U_k$ that satisfies the inequalities (1) and (2). Now by the way that we construct schedules and by Lemma~\ref{claim-delta-close-sumcj}, 
in $U_{k+1}$,  there must be  a  schedule $S_{k+1}^{\delta} \in U_{k+1}$ that
is the same as $S_k$ for the jobs from groups $k_0$ to $k$,  and is
$\delta_{k+1}$-close to  $S_{k+1}^*$ for the jobs of group $(k+1)$.

Then for each machine $M_i$ we have 
\begin{align*} 
P_i(S^{\delta}_{k+1}) & = 
             P_i(S_k)+P_{i, k+1}(S^{\delta}_{k+1}) \\
& \le (1 + \delta)^{k-k_0+2}
P_i(S^*_k) + (1+\delta)P_{i, k+1}(S^*_{k+1}) \\
& \le (1 + \delta)^{k-k_0+2} P_i(S^*_{k+1})
\end{align*} 
And compared with $S^*_{k+1}$, on each machine
$M_i$, the first job from group $(k+1)$  group in $S^{\delta}_{k+1}$ is
delayed by at most $\frac{P_i(S_k)-P_i(S^*_{k})}{\alpha_0} \le
\frac{(1+\delta)^{k-k_0+2}-1}{\alpha_0}\cdot P_i(S^*_{k})$. By Lemma
\ref{claim-same-length-job} Case (4), for the jobs from group $k+1$ on
each machine $M_i$, we have $\sigma_{i, k+1}(S^{\delta}_{k+1}) \le
(1+\frac{(1+\delta)^{k-k_0+2}-1}{\alpha_0})( 1 +
\tfrac{3\delta}{\alpha_0}) \sigma_{i, k+1}(S^*_{k+1})$ and

     \begin{eqnarray*}
     & &  \sigma_{i}(S^{\delta}_{k+1}) \\
      &  = &
      \sigma_{i}(S_{k}) +  \sigma_{i, k+1}(S^{\delta}_{k+1}) \\
      &  \le &
        \sigma_{i}(S_{k}) +  (1+\tfrac{(1+\delta)^{k-k_0+2}-1}{\alpha_0})( 1 + \tfrac{3\delta}{\alpha_0}) \cdot \sigma_{i, k+1}(S^*_{k+1}) \\
      &  \le &
        \sigma_{i}(S_{k}) +  (1+\tfrac{(1+\delta)^{\mu}-1}{\alpha_0})( 1 + \tfrac{3\delta}{\alpha_0}) \cdot \sigma_{i, k+1}(S^*_{k+1}) \\
        &  \le &
         \sigma_{i}(S_{k}) +  (1+\tfrac{2\mu\delta}{\alpha_0})( 1 + \tfrac{3\delta}{\alpha_0}) \cdot \sigma_{i, k+1}(S^*_{k+1}) \\
         & \le & (1 + \delta)^{k-k_0+1} (1+ \tfrac{ 2 \mu \delta}{\alpha_0})(1 + \tfrac{3\delta}{\alpha_0})   (\sigma_i(S^*_k) +  (1+\tfrac{2\mu\delta}{\alpha_0})( 1 + \tfrac{3\delta}{\alpha_0}) \cdot \sigma_{i, k+1}(S^*_{k+1})\\
         &  \le &
        (1 + \delta)^{k-k_0+1} (1+ \tfrac{ 2 \mu \delta}{\alpha_0})(1 + \tfrac{3\delta}{\alpha_0}) \cdot (\sigma_i(S^*_k) +  \sigma_{i, k+1}(S^*_{k+1})) \\
      &  \le &
        (1 + \delta)^{k-k_0+1} (1+ \tfrac{ 2 \mu \delta}{\alpha_0})(1 + \tfrac{3\delta}{\alpha_0}) \cdot \sigma_i(S^*_{k+1})
     \end{eqnarray*}

     Then after the ``similar'' schedules are pruned in our procedure, there is a schedule $S_{k+1}$ that is ``similar'' to $S^{\delta}_{k+1}$, so for each machines $M_i (1 \le i \le m)$ we have 
     $$P_i(S_{k+1}) \le (1 + \delta) P_i(S^{\delta}_{k+1}) \le (1 + \delta)^{(k+1)-k_0+2} P_i(S^*_{k+1})$$ and $$\sigma_i(S_{k+1}) \le (1 + \delta)\sigma_i(S^{\delta}_{k+1}) \le (1 + \delta)^{(k+1)-k_0+1} (1+ \tfrac{ 2 \mu \delta}{\alpha_0})(1 + \tfrac{3\delta}{\alpha_0}) \sigma_i(S^*_{k+1}).$$
     This completes the proof.
\end{proof}

After all groups of jobs are scheduled, our algorithm finds the schedule $S'$ that has the smallest total completion time among all generated schedules, and then returns the value   $(1 + {\epsilon}/{3}) (1+{\epsilon}/{15})\sigma(S')$. 
In the following we will show that the returned value is an approximate value of the optimal total completion time for the  job set $N$.

\begin{lemma}\label{lemma-close-approx}
Let $S^*$ be  the optimal schedule for jobs in $N$,  $(1 + {\epsilon}/{3} ) (1+\epsilon /15) \sigma (S') \le   (1 + \epsilon)  \sigma(S^*) $.
\end{lemma}
\begin{proof}
We first construct a schedule of jobs in $N$ based on the schedule $S'$ of jobs in the sketch $N'_L$ using the following two steps:   
\begin{enumerate}
\item Replace each job in $S'$ with the corresponding job from $N$, let the schedule be $S''$
\item Insert all small jobs from $N \setminus N_L$ to $M_1$ starting at time 0 to $S''$, and let the schedule be $S$
\end{enumerate}

For each job $j$ with processing time $(1+\tau)^{k-1} \le p_j < (1+\tau)^k$, its rounded processing time is $rp_k=\floor{(1+\tau)^k} \ge p_j$, so when we replace $rp_k$ with $p_j$ to get $S''$,   the completion time of each job will not increase, and thus the total completion time of jobs in $S''$ is at most that of $S'$. That is, $\sigma(S'') \le \sigma(S')$.

 All the small jobs  have processing time at most  $ \tfrac{\epsilon \alpha_0}{3n^2} p_{max}$, so the total length is at most  $n \cdot  \tfrac{\epsilon \alpha_0}{3n^2} p_{max} $. Inserting them to $M_1$, the completion time of the last small job  in $S$ is at most $n \cdot p_{max} \tfrac{\epsilon \alpha_0}{3n^2} \cdot \tfrac{1}{\alpha_0}$, and the other jobs' completion time is increased by at most $n \cdot \tfrac{\epsilon \alpha_0}{3n^2} p_{max} \cdot \tfrac{1}{\alpha_0}$. The total completion time of all the jobs is at most
$$\sigma(S) \le  \sigma(S'') + n^2 \cdot  \tfrac{\epsilon \alpha_0}{3n^2} p_{max} \cdot \tfrac{1}{\alpha_0}   \le  \sigma(S'') + \tfrac{\epsilon}{3} \sigma(S'') \le (1 + \tfrac{\epsilon}{3})\sigma(S'') \le (1 + \tfrac{\epsilon}{3}) \sigma(S') .$$ 

By Lemma~\ref{lemma-close-approx-after-prune},  there is a schedule $S_{k_1} 
 \in U_{k_1}$  that corresponds to the schedule of the large jobs obtained from $S^*$.  Furthermore, 
 $\sigma(S_{{k_1}}) \le (1 + \delta)^{\mu} (1+ \tfrac{ 2 \mu \delta}{\alpha_0})(1 + \tfrac{3\delta}{\alpha_0}) \sigma(S^*_{k_1})$ where $\mu=k_1-k_0+1 = \log_{1 + \tau} \tfrac{3 n^2}{\epsilon \alpha_0}$.  For  schedule $S'$,  we have $\sigma(S') \le \sigma(S_{{k_1}})$. Thus,
\begin{eqnarray*}
      \sigma(S) &  \le &
     (1 + \tfrac{\epsilon}{3})\sigma(S')
     \\
      &  \le &
       (1 + \tfrac{\epsilon}{3})(1 + \delta)^{\mu} (1+ \tfrac{ 2 \mu \delta}{\alpha_0})(1 + \tfrac{3\delta}{\alpha_0}) \sigma(S^*_{k_1}) \\
        &  \le &
        (1 + \tfrac{\epsilon}{3})(1+ \tfrac{ 2 \mu \delta}{\alpha_0}) (1+ \tfrac{ 2 \mu \delta}{\alpha_0})(1 + \tfrac{3\delta}{\alpha_0}) \sigma(S^*) \\
         &  \le &
        (1 + \tfrac{\epsilon}{3})(1+ \tfrac{\epsilon}{12})^2(1 + \tfrac{3\epsilon}{24\mu}) \sigma(S^*)    \hspace {0.2in}  \text{plug in }\delta = \tfrac{\epsilon\alpha_0}{  24 \mu} \text{ and } \mu = \log_{1 + \tau} \tfrac{3 n^2}{\epsilon \alpha_0}\\
       &  \le &
        (1 + \tfrac{\epsilon}{3})(1+ \tfrac{\epsilon}{12})^2(1 + \tfrac{\epsilon}{8 }) \sigma(S^*) \\
      &  \le &
        (1 + \epsilon)) \cdot \sigma(S^*)
     \end{eqnarray*}
\end{proof}

\begin{lemma}\label{lemma-delta-close-approx-time-space} The Algorithm   in stage 2 runs in 
$$O( (k_1-k_0 + 1) (m{\log_{1+\delta} n } )^{(m-1)} ({\log_{1+\delta} \tfrac{\sum p_j}{\alpha_0} } )^{m}({\log_{1+\delta} (n \tfrac{\sum p_j}{\alpha_0}) } )^{m}) $$ 
time using  
$O(({\log_{1+\delta}  \tfrac{\sum p_j}{\alpha_0}  } )^{m}({\log_{1+\delta} (n \tfrac{\sum p_j}{\alpha_0}) } )^{m})) $ space.
\end{lemma}
\begin{proof}
For each rounded processing time $rp_k$ and for each of $(m-1)$ machines, the number of possible jobs assigned  is either $0$, or $\floor{(1 + \delta)^l}$, $1 \le l \le {\log_{1+\delta} n_k }$, so there are $ O({\log_{1+\delta} n })$ values. The remaining jobs will be assigned to the last machine. There are at most $O( m ({\log_{1+\delta} n }   )^{m-1})$ ways to assign the jobs of a group to the $m$ machines.

For each schedule in $U_k$, the largest completion time  is  bounded by
$ L = \sum_{1   \le j \le n} p_j / \alpha_0$, and its total completion time is bounded by $ n L$.
Since we only keep non-similar schedules, there are at most $O((\log_{1+\delta}{L})^{m}(\log_{1+\delta}{n L})^{m})$ schedules.
\end{proof}

Combining Lemma~\ref{lemma:computing-sketch}, Lemma~\ref{lemma-close-approx}, and Lemma~\ref{lemma-delta-close-approx-time-space}, we get the following theorem.

\begin{theorem}\label{approximation-scheme-streaming-thm2}Let $\alpha_0$ and $\epsilon$ be a real in $(0,1]$. For $P_m, \alpha_{i,k} \ge \alpha_0 \mid stream \mid \sum C_j$,  there is  a 
one-pass $(1+\epsilon)$-approximation scheme with the 
following time and space complexity:
    \begin{enumerate}
        \item Given both an upper bound $n'$ for the number of jobs $n$  and a lower bound $p'_{max}$ for the largest processing time of job $p_{max}$ such that   $n\le n'\le c_1 n$ and   $1\le p_{max}'\le p_{max}\le  c_2p_{max}'$ for some $c_1$ and $c_2$,
        then it takes $O(1)$  update time and
         $O(\tfrac{1}{ \epsilon\alpha_0}\min (\log n+\log \tfrac {c_1c_2} {\epsilon\alpha_0},\log p_{max}+\log c_2))$ space to process each job from the stream. 
      \item Given only a lower bound $p'_{max}$ for $p_{max}$ where  $1\le p_{max}'\le p_{max}\le  c_2p_{max}'$,  then it  takes  $O(1)$  update time  and 
         $O(\tfrac{1}{\epsilon\alpha_0}(\log p_{max}+\log c_2))$ space to process each job in the stream.
    \item Given only  an upper bound $n'$ for $n$ such that   $n\le n'\le c_1 n$,  then   it  takes  $O(\log (\tfrac{1}{ \epsilon\alpha_0}) + \min( \log (\log n+\log \tfrac{c_1}{  \epsilon\alpha_0})), \log \log p_{max})$  update time, and 
         $O(\tfrac{1}{\epsilon\alpha_0}\min(\log n+\log \tfrac{c_1}{  \epsilon\alpha_0},\log p_{max}))$ space to process each job in the stream.
     \item Given no information about $n$ and $p_{max}$, then it takes  $O(\log (\tfrac{1}{ \epsilon\alpha_0} + \log \log p_{max}))$  update time, and      $O(\tfrac{1}{ \epsilon\alpha_0}\log p_{max})$ space to process each job in the stream.
        \item After processing the input stream, to compute the approximation value, it  takes 
%$$O( (\tfrac{1}{\epsilon \alpha_0} \log(\tfrac{ n}{\epsilon \alpha_0 } (m{\log_{1+\delta} n } )^{(m-1)}) ({\log_{1+\delta} \tfrac{\sum p_j}{\alpha_0} } )^{m}({\log_{1+\delta} (n \tfrac{\sum p_j}{\alpha_0}) } )^{m})$$
$$O({\epsilon\alpha_0} \cdot (\tfrac{360}{\epsilon^2\alpha_0^2} \log \tfrac{n}{\epsilon\alpha_0})^{3m} \cdot (m \log n)^{m-1} \cdot (\log \tfrac{\sum p_j}{\alpha_0} \cdot \log \tfrac{n\sum p_j}{\alpha_0})^m)$$
time using  
$O((\log \tfrac{n}{\epsilon \alpha_0})^{2m} ({\log ( \tfrac{\sum p_j}{\alpha_0})  \log (n \tfrac{\sum p_j}{\alpha_0}) } )^{m})$ .
    \end{enumerate}
\end{theorem}

Note that our algorithm only finds an approximate value for the optimal total completion time, and it does not generate the schedule of all jobs. 
If the jobs can be read  in a second pass, we can return a schedule of all jobs whose total completion time is at most $(1 + \epsilon)\sigma(S^*)$ where $S^*$ is an optimal schedule.
Specifically, after the first pass, we store $n_{i,k}(S')$, $1 \le i \le m$, $k_0 \le k \le k_1$, which is the number of large jobs from  group $k$ that are assigned to machine $M_i$ in $S'$. Based on the selected schedule $S'$, we get $t_{i,k}$ that is the starting time for jobs from group $k$, $k_0 
\le k \le k_1$, on each machine $M_i$, $1 \le i \le m$. We add group $k_0-1$ that includes all the small jobs and will be scheduled at the beginning on machine $M_1$, that is, initially $t_{1,k_0-1}=0$. For all $t_{1,k}$, $k_0 \le k \le k_1$, we update it by adding $n \tfrac{\epsilon \alpha_0}{3n^2}\tfrac{p_{max}}{\alpha_0}$.  
In the second pass, for each job $j$ scanned, if it is a large job and its rounded processing time is $rp_k$ for some $k_0 
\le k \le k_1$, we schedule it to a machine $M_i$ with $n_{i,k}(S') > 0$ at $t_{i,k}$  and then update $n_{i,k}(S')$ by decreasing by 1 and update $t_{i,k}$ accordingly; otherwise, job $j$ is a small job, and we schedule this job at $t_{1, k_0-1}$ on machine $M_1$ and update $t_{1, k_0-1}$ accordingly. The total space needed in the second pass is for storing $t_{i,k}$ and $n_{i,k}(S')$ for $1 \le i \le m$ and $k_0 \le k \le k_1$, which is $O(m(k_1-k_0+1))=O(\tfrac{m}{\epsilon \alpha_0} \log \tfrac{n}{\epsilon \alpha_0})$.

\begin{theorem}
There is a two-pass $(1+\epsilon)$-approximation streaming algorithm for $P_m, \alpha_{i,k} \ge \alpha_0 \mid stream \mid \sum C_j$. In the first pass, the approximate value can be obtained with the  same metrics as Theorem \ref{approximation-scheme-streaming-thm2}; and in the second pass, a schedule with the approximate value can be returned with $O(1)$ processing time and $O(\tfrac{m}{\epsilon \alpha_0} \log \tfrac{n}{\epsilon \alpha_0})$ space for each job.
\end{theorem}

\section{Conclusions}

In this paper we studied a generalization of the classical identical parallel machine scheduling model, where the processing capacity of machines varies over time.
This model is motivated by situations in which machine availability is temporarily reduced to conserve energy or interrupted for scheduled maintenance or varies over time due to the varying labor availability. The goal is to minimize the total completion time.

We  studied the problem under the data stream model and presented
the first streaming algorithm. Our work follows the study of
streaming algorithms in the area of statistics, graph theory, etc,
and leads the research direction of streaming algorithms in the area
of scheduling. It is expected that more streaming algorithms based
big data solutions will be developed in the future. 

Our research leaves
one unsolved case for the studied problems: Is there a streaming
approximation scheme when one of the machines has arbitrary
processing capacities?  For the future work, it is also interesting
to study other performance criteria under the data stream model including maximum tardiness, the
number of tardy jobs and other machine environments such as uniform
machines, flowshop, etc.

\bibliography{bibliography}
\end{document}